\documentclass[11pt]{article}
\usepackage{fullpage}
\newenvironment{proof}{\noindent\textbf{Proof.}\quad}{~\qedsymbol}

\usepackage{eepic,epic}
\usepackage{epsfig}
\usepackage{graphicx}
\usepackage{color}

\usepackage{amssymb}
\usepackage{amsmath}
\usepackage{pifont}



\def\literalqed{{\ \nolinebreak\hfill\mbox{\qedblob\quad}}}

\makeatletter \@beginparpenalty=10000 \makeatother

\mathcode`\0="0030      %
\mathcode`\1="0031
\mathcode`\2="0032
\mathcode`\3="0033
\mathcode`\4="0034
\mathcode`\5="0035
\mathcode`\6="0036
\mathcode`\7="0037
\mathcode`\8="0038
\mathcode`\9="0039

\newcommand\qedblob{\mbox{\ding{113}}}
\def\qedsymbol{{\ \nolinebreak\hfill\mbox{\qedblob\quad}}\smallskip}

\newcommand{\naturals}{\mathbb{N}}
\newcommand{\integers}{\mathbb{Z}}

\newcommand{\capacity}{\mathit{cpc}}
\newcommand{\price}{\mathit{prc}}

\newcommand{\kbelection}{\ensuremath{(k,b)\hbox{-}\mathrm{election}}}

\newtheorem{theorem}{Theorem}[section]

\newtheorem{corollary}[theorem]{Corollary}
\newtheorem{definition}[theorem]{Definition}

\newcommand{\np}{\ensuremath{\mathrm{NP}}}

\begin{document}

\title{Nonuniform Bribery\thanks{%
Supported in part by grant NSF-CCF-0426761. A version of this paper appears as URCS-TR-2007-922.}
}

\date{November 30, 2007}

\author{Piotr Faliszewski \\ Department of Computer Science \\ University of Rochester \\ Rochester, NY 14627}

\maketitle
\begin{abstract}
  We study the concept of bribery in the situation where voters are
  willing to change their votes as we ask them, but where their prices
  depend on the nature of the change we request. Our model is an
  extension of the one of Faliszewski et
  al.~\cite{fal-hem-hem:c:bribery}, where each voter has a single
  price for any change we may ask for.  We show polynomial-time
  algorithms for our version of bribery for a broad range of voting
  protocols, including plurality, veto, approval, and utility based
  voting. In addition to our polynomial-time algorithms we provide
  $\np$-completeness results for a couple of our nonuniform bribery
  problems for weighted voters, and a couple of approximation
  algorithms for $\np$-complete bribery problems defined
  in~\cite{fal-hem-hem:c:bribery} (in particular, an FPTAS for
  plurality-weighted-\$bribery problem).
\end{abstract}

\section{Introduction}

Multiagent systems can often be viewed as artificial societies of
autonomous agents, with each agent having his/her own set of goals,
desires, and plans. Within such artificial societies, just like within
natural ones, it often becomes necessary for a group of agents to
arrive at a common decision (e.g., in a planning environment when no
single agent can solve his/her own problems, but where they are
capable of solving problems cooperatively). A very natural approach to
handling such situations is to hold an election.

Unfortunately, as is well known due to theorems of
Arrow~\cite{arr:b:polsci:social-choice}, Gibbard and
Satterthwaite~\cite{gib:j:polsci:manipulation,sat:j:polsci:manipulation},
and Duggan and Schwartz~\cite{dug-sch:j:polsci:gibbard} neither there
are ideal election systems nor there are ones that avoid giving agents
incentive to vote strategically. Bartholdi, Tovey, Trick, and
Orlin~\cite{bar-tov-tri:j:manipulating,bar-oli:j:polsci:strategic-voting,bar-tov-tri:j:control},
brilliantly observed that computational complexity of figuring out
voter's strategic behavior might be so high, as to perhaps be
enough of a barrier to prevent agents from attempting strategic actions.

Several scenarios of strategic behavior are considered in the
literature. In control we assume that the organizer of the election
attempts to modify their structure (e.g., via partitioning the voters
into districts) in order to obtain a result most desirable for
him/herself, see,
e.g.,~\cite{bar-tov-tri:j:control,fal-hem-hem-rot:c:llull,hem-hem-rot:j:destructive-control,pro-ros-zoh:c:multiwinner}.
In the case of manipulation, a coalition of voters calculates what vote
each member of the coalition should cast in order to obtain the result
the coalition desires, see,
e.g.,~\cite{bar-tov-tri:j:manipulating,bar-oli:j:polsci:strategic-voting,con-lan-san:j:few-candidates,con-san:c:voting-tweaks,con-san:c:nonexistence,hem-hem:j:dichotomy,pro-ros:j:juntas,pro-ros-zoh:c:multiwinner}

In the case of bribery, an external agent tries to ensure a victory of
one of the candidates via bribing some voters to change their
votes. Bribery was introduced in~\cite{fal-hem-hem:c:bribery} and
further studied in \cite{fal-hem-hem-rot:c:llull}.

All of these problems can be studied both in the constructive case
(where the goal is to ensure our favorite candidate's victory) and in
the destructive case (where we try to prevent a hated candidate from
winning). Also, in many situations it is natural to assume that
different voters have different weights (e.g., consider the
stockholders of a company or various parts of a multicriteria
decision-making system that, internally, performs an election between
its components, weighted based on components' confidence.)

In this paper we focus on the problem of bribery, introduced by
Faliszewski et al.~\cite{fal-hem-hem:c:bribery}. The authors of that
paper studied bribery in several scenarios, depending on the voting
system used and whether the voters were weighted and/or were assigned
price-tags for changing their votes. In particular, in the priced
cases they assumed that each voter is willing to change his/her vote
arbitrarily, provided that the briber pays the fixed-per-voter
price. Such an assumption is fairly unrealistic. For example, let us
consider an election with four candidates, $a$, $b$, $c$, and $d$. It
is easy to imagine a voter who, e.g., prefers option $a$ to $b$ to $c$
to $d$, but who actually really wants either $a$ or $b$ to win, with a
slight preference towards $a$, and who absolutely does not want either
of the remaining options to be chosen. Such a voter might be willing,
at a small price, to change his/her vote to rank $b$ first, but would
never, regardless of the money offered, change the vote to rank either
$c$ or $d$ first.

A different example where it is useful to model voters as having such
nonuniform prices is best seen from the point of view of the
briber.\footnote{The following discussion is not very technical; we
  defer technical issues to further sections.}  A briber that wants
some candidate $p$ to win, might want to follow a certain
\emph{policy} in his/her bribing. For example, he/she might not want
to bribe anyone to vote for $p$ in order not to cast ``bad light'' on
$p$. Such a briber would have to make $p$ a winner via bribes that
redistribute other candidate's support. Using the nonuniform model of
bribery one could express this policy via setting the prices for
voting for $p$ so high as to be outside of the allowed budget.

The above policy was studied in~\cite{fal-hem-hem:c:bribery} but,
naturally, different policies can easily be devised. (E.g., policies
preventing the briber from affecting some voters, or only allowing
him/her to bribe voters to change their votes in a limited way.)

Yet another scenario where nonuniform bribery model is useful regards
the issue of coalition formation. Consider an election where one of
the voters realizes that his/her option is very unlikely to win, but
where there are many agents (voters) that support options similar, but
slightly different. Such an agent might want to find out which of the
others, but as few as possible, he/she would have to convince to form
a coalition with him/her in order to have enough voting power as to
choose an option that all of them would be reasonably satisfied
with. One way to compute this set would be to: (a) find a group of
voters that currently vote for options similar to the agent's, (b)
form nonuniform bribery instance where those voters can be bribed at a
relatively low price to vote for the agent's option and all other
briberies are either very expensive or impossible (beyond budget), (c)
compute a minimum-cost nonuniform bribery that ensures that the
agent's favorite option wins. The voters involved in this bribery
would be candidates for the coalition.

In view of the above discussion we see that the issue of nonuniform
bribery is both important and useful, even in the cases where we are
not really ``bribing'' anyone, but simply are trying to strategically
plan our behavior. In this paper we give a number of results that show
that the problem of nonuniform bribery can often be solved in
polynomial time in the case of unweighted voters, yet easily becomes
$\np$-complete if the voters are weighted. We also give several
approximation algorithms for previously studied $\np$-complete bribery
problems.

\section{Preliminaries}
In this section we provide our notation, model of elections, and
briefly describe our main algorithmic tool, flow networks.

We view an election as a pair $(C,V)$, where $C$ is a set of
candidates, $C = \{c_1, \ldots, c_m\}$, and $V$ is a multiset of
voters, each represented via his/her preference over $C$. In this
paper we represent each voter's preference as an $m$-dimensional
vector of nonnegative integers, indicating voter's perceived utility
from electing each candidate. (Our model is inspired by the discussion 
of a similar notion in the paper by Elkind and Lipmaa~\cite{elk-lip:c:hybrid-manipulation}.)

\begin{definition}
  \label{def:kb-election}
  Let $k$ and $b$ be two positive integers. By a $\kbelection$ we mean
  an election over some candidate set $C = \{c_1, \ldots, c_m\}$,
  where each voter from the voter set $V$ distributes $k$ integral
  points among the candidates, never assigning more than $b$ points to
  a single candidate.\footnote{Note that if either $k$ is too big, or
    there aren't enough candidates then such an election is
    impossible.} In the end, the candidates with most points are the
  winners.

  A free-form $\kbelection$ is a $\kbelection$ where voters can choose
  not to use all of their points.
\end{definition}

The standard preference model studied in computational social choice
literature assumes that each voter has a strict linear order of
preference over all candidates. Our model is more appropriate
for utility-based voting, and is powerful enough to capture such
voting rules as plurality (as $(1,1)$-elections; in plurality
elections each voter assigns one point to his/her favorite candidate),
veto (as $(m-1,1)$-elections; in veto election each voter gives a
single point to everyone except for his/her most hated option),
approval (as free-form $(m,1)$-elections; in approval voting each
voter either approves or disapproves of each of the candidates and the
candidates with most approvals win), and $t$-approval (as
$(t,1)$-elections; as approval, but each candidate has to approve of
exactly $t$ candidates). Utility-based voting is a very attractive
concept and we believe it is interesting to study it in computational
context.

We use $(k,b)$-elections as our model of voting. It is particularly
nice for our nonuniform bribery as one can devise a reasonable pricing
scheme for voters; we assume that for each ordered pair of candidates
each voter has a, possibly different, price for moving a single point
from one candidate to the other. Similar schemes can be devised for
preference-order-based voting but those we came up with were neither
as convincing nor as elegant.

The main result of this paper, Theorem~\ref{thm:main}, follows via an
algorithm that uses min-cost flow problem as a crucial subroutine. We
now briefly describe the min-cost flow problem and define appropriate
notation.

A flow network is defined via a set of nodes $N = \{s, t, n_1, \ldots,
n_m\}$, where $s$ is the source and $t$ is the sink, a capacity
function $\capacity: N \times N \rightarrow \naturals$, and a cost
function $\price: N \times N \rightarrow \naturals$. We say that two
nodes, call them $u$ and $v$, are connected if $\capacity(u,v) >
0$. Note that $\capacity(u,v)$ does not need to equal
$\capacity(v,u)$; the channels connecting two nodes are
unidirectional. A function $f: N \times N \rightarrow \integers$ is a
flow in such a network if it satisfies the following constraints: (a)
$(\forall_{u,v \in N})[f(u,v) \leq \capacity(u,v)]$ (i.e., we do not
send flow beyond capacity), (b) $(\forall_{u,v\in N})[f(u,v) =
-f(v,u)]$, and (c) $(\forall_{u \in N - \{s,t\}})[\sum_{v \in N
  -\{u\}}f(u,v) = 0]$.  We interpret $f(u,v) = t$, $t > 0$, as $t$
units of flow traveling on the edge from $u$ to $v$. A
negative flow value indicates reversed direction of travel. Note that
units of flow can travel in both directions at the same time, but it
is unnecessary.  The value of a flow is defined as $\sum_{u \in N -
  \{s\}}f(s,u)$ and its cost is $\sum_{u,v \in N \mid f(u,v)\geq 0}
\price(u,v)f(u,v)$. That is, for each two nodes $u,v \in N$ we pay
$\price(u,v)$ for sending each single unit of flow from $u$ to $v$.

In the min-cost flow problem we are given a nonnegative integer $K$, a
flow network, i.e., a set of nodes with source and sink nodes, a
capacity function, and a cost function, and our goal is to find a
minimum-cost flow of value $K$ (or to indicate that such a flow
doesn't exist). It is well known that this problem is solvable in
polynomial-time and we point interested readers to an excellent
monograph of Ahuja et al.~\cite{ahu-mag-orl:b:flows}.

Finally, we explain what we mean by a fully polynomial-time
approximation scheme (FPTAS). Due to space constraints the definition
is very brief. An algorithm $A$ is an FPTAS for a given minimization
problem (e.g., a problem of finding minimum-cost bribery) if (a)
algorithm $A$ on input $(I,\varepsilon)$, where $I$ is an instance of
the problem and $\varepsilon$ is a positive real number, $0 <
\varepsilon < 1$, runs in time polynomial in the size of $I$ and
$\frac{1}{\varepsilon}$, and (b) $A$ has the property that if $O$ is
the value of the optimal solution for $I$ then $A(I,\varepsilon)$
produces a solution with value $S$ such that $S \leq (1 +
\varepsilon)O$. We stress that algorithm $A$ has to output a correct
solution for the problem, only that the cost of this solution may be
somewhat above the optimum.

\section{Nonuniform Bribery}
In this section we define our nonuniform bribery problem for
$(k,b)$-elections and present our results.

Let $k$ and $b$ be two positive integer values (possibly dependent on
the number of candidates and voters in $E$; see the next sentence). We
define $(k,b)$-bribery problem as follows: The input is a
$(k,b)$-election $E$ with candidate set $C = \{c_1,\ldots, c_m\}$ and
a multiset $V$ of voters $v_1, \ldots, v_n$, where each voter $v_i$ is
represented via an $m$-dimensional integer vector describing in an
obvious way how many points $v_i$ assigns to which candidates, a
nonnegative integer $B$ (the budget), and for each voter $v_i$ a price
function
\[\pi_i: C \times C \rightarrow \naturals.\]
A unit bribery involves asking some voter $v_\ell \in V$ to move a
single point that $v_\ell$ currently assigns to some candidate $c_i$
to another candidate $c_j$.  The cost of such a unit bribery is
$\pi_\ell(c_i,c_j)$. (Naturally, for each $\ell \in \{1,\ldots, n\}$
and each candidate $c_i$ we have $\pi_\ell(c_i,c_i) = 0$.) The
question is if given the budget $B$, this $(k,b)$-election $E$, and
functions $\pi_1, \ldots, \pi_n$ (one for each voter), it is possible
to perform a set of unit briberies of total cost at most $B$, such
that (a) preferred candidate $p = c_1$ becomes a winner, and, (b) each
voter assigns at most $b$ points to each candidate (i.e., the rules of
a $(k,b)$-election are not broken). We explicitly require that all the
unit briberies are executed ``in parallel,'' that is, a briber cannot
first bribe voter $v_\ell$ to move a point from some candidate $c_i$
to another candidate $c_j$ and afterward move that same point from
$c_j$ to yet another candidate $c_q$.\footnote{We could allow such
  sequential operations and all our results would hold via proofs
  similar to those presented here, but we feel that the model of
  parallel unit briberies is more appropriate.} (Though, it would
still be legal to move to $c_q$ a point that $v_\ell$ had assigned to
$c_j$ before the bribery. Of course, points are anonymous, unnamed,
entities; our requirement of not moving ``the same'' point twice
formally means that briberies are only legal if for each voter
$v_\ell$ they move, within the preference vector of that voter, at
most as many points away from each candidate as that candidate
had been assigned by $v_\ell$ before bribery.)

We define free-form $(k,b)$-bribery problem analogously, only that
each voter can choose to not assign some of his/her points to any
candidate and we can bribe voters to either give those unassigned
points to some candidate, or to remove points from a given candidate
and not assign them to anyone else. In order to accommodate this
possibility we extend the price functions $\pi_\ell$ to be mappings
from $\hat{C} \times \hat{C}$ to $\naturals$, where $\hat{C} = C \cup
\{\varepsilon\}$. $\varepsilon$ represents the slot for unassigned
points. Naturally, each voter can have as many unassigned points as
he/she wishes to (i.e., the $b$-bound does not apply to the unassigned
points).

We now give our main result regarding $(k,b)$-bribery.

\begin{theorem}
  \label{thm:main}
  There is an algorithm that solves both $(k,b)$-bribery instances and
  free-form $(k,b)$-bribery instances in time polynomial in $k$ and
  the size of the instance.
\end{theorem}

We devote a large chunk of the remainder of this section to proving
Theorem~\ref{thm:main}. However, some discussion is in order before we
jump into the proof. In particular, note that the running time of our
algorithm is polynomial not in the size of $(k,b)$-bribery instance,
but in the size of the instance \emph{and} the value $k$.  This means
that if the value of $k$ is large then the algorithm runs very
slowly. However, for many interesting cases (all that we have
mentioned, e.g., plurality, veto, utility based voting where $k$ is
bounded by some polynomial in the number of candidates) our result
guarantees polynomial-time solution for this nonuniform bribery
problem.

\noindent\textbf{Proof of Theorem~\ref{thm:main}}\quad
  Due to length restrictions, we only provide a proof sketch for the
  case of $(k,b)$-bribery, and not for the free-form variant.

  Our input is a $(k,b)$-election $E$, with candidate set $C = \{c_1,
  \ldots, c_m\}$ and voter multiset $V = \{v_1, \ldots, v_n\}$, a
  nonnegative integer $B$  (the budget), and voters' price functions
  $\pi_1, \ldots, \pi_n$. Our goal is to ensure that our distinguished
  candidate $p = c_1$ is a winner of the election via a bribery of
  cost at most $B$.

  Our proof follows via constructing a series of flow networks and
  computing min-cost solutions for them. The intuition here is that
  the points that voters assign to candidates are modeled via the
  units of flow traveling in the network. We design our networks in
  such a way that minimizing the cost of the flow, in essence,
  maximizes the number of points our designated candidate $p = c_1$
  via a minimum-cost bribery.

  We know that each candidate can at most receive $kn$ points. For
  each nonnegative integer $K$ between $1$ and $kn$ our algorithm
  tests if there is a bribery of cost at most $B$ that ensures that
  $p$ receives exactly $K$ points and every other candidate receives
  at most $K$ points. Let us now fix a value of $K$ and show how such
  a test can be executed.

  We form a network flow with node set $N = \{s,t\} \cup
  \bigcup_{i=1}^n C_i \cup \bigcup_{i=1}^n C'_i \cup F$, where $F =
  \{f_1, \ldots, f_m\}$ and for each $i \in \{1,\ldots, n\}$ we have
  $C_i = \{c_{i1}, \ldots, c_{im}\}$, $C'_i = \{c'_{i1}, \ldots,
  c'_{im}\}$. For each $i \in \{1,\ldots ,n\}$, nodes in the sets
  $C_i$ represent point distribution of voter $v_i$ before bribery,
  nodes in $C'_i$ represent point distribution of voter $v_i$ after
  the bribery, and nodes in $F$ are used to enforce the rules of
  $(k,b)$-election and to sum points for all candidates.

  We introduce the following capacities and costs for edges in our
  network. Assume that all unmentioned edges have capacity $0$. For
  each voter $v_\ell$ and candidate $c_i$ we have $\capacity(s,c_{\ell
    i})$ equal to the number of points $v_\ell$ assigns to $c_i$
  before bribery and $\price(s,c_{\ell i}) = 0$.  These edges model
  delivering appropriate number of points to nodes from sets $C_1,
  \ldots, C_n$.

  For each node $c_{\ell i}$ and each candidate $c_j$ we have
  $\capacity(c_{\ell i},c'_{\ell j}) = k$ and $\price(c_{\ell i},c_{\ell
    j}) = \pi_\ell(i,j)$. These edges model unit briberies. The briber
  can ask each voter to move his points as the briber likes, but pays
  appropriate price for moving each point.

  For each node $c'_{\ell i}$ we set $\capacity(c'_{\ell i},f_i) = b$ and
  $\price(c'_{\ell,f_i}) = 0$. These edges enforce that no candidate can
  receive more than $b$ points from a single voter. Finally, for each
  node $f_i$, we have $\capacity(f_i,t) = K$, $\price(f_1,t) = 0$, and for all $i
  \in \{2, \ldots, m\}$ we have $\price(f_i,t) = T$, where $T$ is an
  integer higher than the cost of any possible bribery (e.g., take $T
  = 1+kn\max_{\ell,i,j}\pi_\ell(c_i,c_j)$).

  To perform our test we compute minimum-cost flow of value $kn$ in
  this network.  If such a flow doesn't exist then it means that there
  is no legal way of distributing points via bribery in such a way
  that each voter has score at most $K$. In such a case we disregard
  this value of $K$ and continue with the next one.  We can interpret
  our minimum-cost flow as follows: Each set of nodes $C_1, \ldots,
  C_n$ receives units of flow corresponding to the distribution of
  points of each voter (because flow has value $kn$ and because of the
  capacities of edges connecting the source with nodes in $C_1,
  \ldots, C_n$) and units of flow travel from nodes in sets $C_i$ to
  nodes in sets $C'_i$, respectively, modeling our bribery. Finally,
  they all accumulate in nodes from set $F$, from where they all reach
  the sink. The nature of edges between nodes in $C'_i$'s and nodes in
  $F$ guarantees that we arrive with a legal distribution of points
  for each voter. The cost of this flow can be expressed as $T\cdot (
  kn - \mbox{$p$'s score after bribery}) +
  \mathrm{cost\hbox{-}of\hbox{-}bribery}$. Since $T$ is chosen to be
  larger than any possible cost of bribery, we know that minimum cost
  enforces that node $f_1$, corresponding to $p = c_1$, receives as
  many units of flow as legally possible. If $p$ can legally receive
  $K$ units of flow then minimum-cost flow delivers this many units;
  the capacity of the edge linking $p = c_1$ with the sink is $K$ so the
  flow cannot deliver more. If the flow cannot legally deliver $K$
  units of flow to $c_1$ then we can safely disregard this network
  (because we have already handled this flow when analyzing smaller
  values of $K$). Interpretation of our flow gives that our bribery
  guarantees that $p$ gets exactly $K$ points and all the other
  candidates receive at most $K$ points each (as each node $f_2, \ldots,
  f_n$ can only deliver $K$ units of flow to the sink).  This is
  achieved via a minimum cost bribery as the cost-of-bribery is the
  remaining part of the cost of our flow, after we consider the
  payment for delivering units of flow to the sink.

  This way we can test in polynomial-time for each $K$ if there is a
  nonuniform bribery of cost at most $B$ that ensures that $p$
  receives exactly $K$ points and all other candidates receive at most
  $K$ points.  Thus, in total, the running time of our algorithm is
  polynomial in the size of our election and $k$.~\qedsymbol

Since we have shown that by proper choice of the values $k$ and $b$ we
can express plurality, veto, approval, $t$-approval, and
utility-based voting (where the number of points to distribute is
bounded by a polynomial in the number of candidates), we have the
following corollary.

\begin{corollary}
  Nonuniform bribery is solvable in polynomial time for the following
  election systems: plurality, veto, approval, $t$-approval,
  utility-based voting where the number of points each voter can
  distribute is polynomial in the number of candidates.
\end{corollary}

It is natural to ask if the above two results hold for the case of
weighted voters. Let $(k,b)$-weighted-bribery and free-form
$(k,b)$-weighted bribery be the analogs of $(k,b)$-bribery and
free-form $(k,b)$-bribery for the case when voters have weights.
Unfortunately, these problems are $\np$-complete even for very
restricted values of $k$ and $b$. In particular, the following result
holds.

\begin{theorem}
  \label{thm:npcom}
  $(1,1)$-weighted-bribery is $\np$-complete.
\end{theorem}
\begin{proof}
  Faliszewski et al.~\cite{fal-hem-hem:c:bribery} showed that what
  they call plurality-weighted-negative-bribery problem is
  $\np$-complete. This problem is defined as follows: Given a weighted
  plurality election $E = (C,V)$ with a distinguished candidate $p$
  and a nonnegative integer $B$, is it possible to pick up to $B$
  voters in $V$ and change their votes so that neither of them ranks
  $p$ first yet $p$ is a winner of the election? It is easy to see
  that this problem reduces to $(1,1)$-weighted-bribery. Given an
  instance of plurality-weighted-negative-bribery we form an instance
  of $(1,1)$-weighted-bribery with the same election $E$ (it is
  trivial to convert preferences expressed as linear orders to
  appropriate vectors in this case), the same budget $B$, and where
  each unit bribery has cost $1$, except for unit briberies that would
  give a point to $p$. Those cost $B+1$. The reader should convince
  him/herself that this reduction works in polynomial time and is
  correct.

  On the other hand, $(1,1)$-weighted-bribery is trivially in
  $\np$.\end{proof}

We could similarly show that $(m-1,1)$-weighted-bribery, where $m$ is
the number of candidates in the election, is $\np$-complete simply via
invoking the fact that veto-weighted-bribery is
$\np$-complete~\cite{fal-hem-hem:c:bribery}.

Together with Theorem~\ref{thm:npcom} this shows that in general we
should not expect efficient algorithms for $(k,b)$-weighted-bribery
that work for all values of $k$ and $b$. However, there might be
interesting algorithms for special cases.

In particular, we focus on the restriction of $(1,1)$-weighted-bribery
to the case where each voter has a single price for moving his/her
point to any of the candidates. This problem is called
plurality-weighted-\$bribery in~\cite{fal-hem-hem:c:bribery} and we
will use this name.  We also consider a restriction of free-form
$(m,1)$-bribery, where $m$ is the number of candidates, to a situation
where each voter has for each candidate a price of flipping the
support for that candidate. This restriction is equivalent to what
in~\cite{fal-hem-hem:c:bribery} is called
approval-weighted-\$bribery$'$. Again, we will use the name
from~\cite{fal-hem-hem:c:bribery}.

The following theorem shows that these two special cases of
(free-form) $(k,b)$-weighted-bribery can be solved efficiently using
approximation algorithms.

\begin{theorem}
  plurality-weighted-\$bribery and approval-weighted-\$bribery$'$ both
  have fully polynomial-time approximation schemes (FPTAS).
\end{theorem}
\begin{proof}
  Faliszewski et al.~\cite{fal-hem-hem:c:bribery} gave polynomial-time
  algorithms for both plurality-weighted-\$bribery and
  approval-weighted-\$bribery$'$ for the case when prices within those
  problems are encoded in unary. Running times of these algorithms are
  polynomial in the size of the problem \emph{and} the value of the
  largest price. We use this fact here to give an FPTAS for both
  problems. In both cases our results follow by an extension of the
  scaling argument used, e.g., in an FPTAS for Knapsack.

  Let $L$ be either plurality-weighted-\$bribery or
  approval-weighted-\$bribery$'$, let $I$ be an instance of $L$ and let
  $\varepsilon$ be a positive real number, $0 < \varepsilon < 1$.  By
  $T$ we mean the highest price occurring in $I$ and we assume that
  $T$ is at least $1$. Otherwise, any solution is optimal.  Also, we
  let $N$ be the total number of prices that occur in $I$. For
  plurality, $N$ is the number of voters and for approval, $N$ is the
  number of voters times the number of candidates. Thus, an upper
  bound on the cost of any bribery in $I$ is $TN$. We give an FPTAS
  for $L$.

  The idea of our algorithm is to scale down the prices so that they
  are polynomially bounded in $N$ and $\frac{1}{\varepsilon}$ and to run
  the polynomial-time algorithm of~\cite{fal-hem-hem:c:bribery} on
  thus modified instance. However, instance $I$ might include some
  voters with very high prices that are not needed in the optimal
  solution. By choosing a scaling factor appropriate for a high price
  range we might essentially lose all the information regarding the
  smaller prices, those that actually participate in the optimal
  solution. Thus, instead of performing one scaling, we perform
  polynomially many of them. We start with a fairly small scaling
  factor and keep on increasing it. With a small scaling factor some
  of the prices within $I$ will be too large, and so we assume that
  the voters with those prices cannot be bribed.

  Our algorithm executes $\lceil \log T \rceil$ iterations. In each
  iteration variable $t$ contains our current guess of an upper bound
  on the largest price used within the optimal solution.  We start
  with $t=1$ and we double it after every iteration.

  Given a value of $t$, an iteration is executed as follows. Set $K =
  \frac{t\varepsilon}{N}$ and construct an instance $I'$ that is
  identical to $I$, only that: (a) each price $p$ such that $p \leq
  t$, is replaced by $\lceil \frac{p}{K} \rceil$, and (b) any higher
  price is replaced by $\frac{1+2\varepsilon}{\varepsilon} N^2 + 1$.
  Note that due to this transformation $I'$ has all prices
  polynomially bounded in $\frac{1}{\varepsilon}$ and $N$. We find an
  optimal solution for $I'$ using the polynomial-time algorithm
  from~\cite{fal-hem-hem:c:bribery}. If the solution has cost
  $\frac{1+2\varepsilon}{\varepsilon}N^2 +1$ or higher then we discard
  it and otherwise, we store it for future use.

  After all the iterations are finished, we return a stored solution
  with the lowest cost. Clearly, we have at least one solution as
  in the last iteration $t \geq T$.

  We claim that this algorithm finds a solution with cost within
  $2\varepsilon$ of the optimal one. Let $O$ be an optimal solution
  and let $u$ be the highest price used within solution $O$. Let us
  consider an iteration of our algorithm with $t$ such that
  $\frac{t}{2} \leq u \leq t$. Let $S$ be our optimal solution to
  instance $I'$ in this iteration.  Since $I$ and $I'$ vary only in
  voter's prices, $S$ and $O$ are valid solutions for both $I$ and
  $I'$.  By $p(S)$ and $p(O)$ we mean the prices of briberies
  specified in $S$ and in $O$, respectively, expressed using prices from
  instance $I$.  By $p'(S)$ and $p'(O)$ we mean analogous values, but
  with respect to prices in $I'$. It holds that
  \[ p(O) \leq p(S) \leq Kp'(S) \leq Kp'(O). \] The first inequality
  holds because $O$ is an optimal solution to $I$ and the second one
  follows because instance $I'$ has prices rounded up.  The last
  inequality is due to the fact that $S$ is an optimal solution for
  $I'$. For any price $p$ in $I$, we have a corresponding price $p' =
  \lceil \frac{p}{K} \rceil$ in $I'$ and, due to rounding, it is easy
  to see that $p \leq Kp' \leq p + K$. Since any bribery involves
  paying at most $N$ prices we have that
  \[ Kp'(O) \leq p(O) + NK.\] Since $NK = \varepsilon t$, and in this
  iteration we have $\frac{t}{2} \leq u \leq t$, naturally we have
  that $\frac{t}{2} \leq u \leq p(O)$. Thus, we have
  \[ p(S) \leq Kp'(S) \leq p(O) + 2\varepsilon p(O) =
  (1+2\varepsilon)p(O). \] It remains to see that our algorithm does
  not discard solution $S$.

  We have that the highest price in $O$ is $u$, and so $p(O) \leq Nu
  \leq Nt$.  Via the above estimate we have that $Kp'(S) \leq p(O)
  (1+2\varepsilon) \leq (1+2\varepsilon)Nt$, and so $p'(S) \leq
  \frac{Nt}{K}(1+2\varepsilon) \leq \frac{1+2\varepsilon}{\varepsilon}
  N^2$. Thus, solution $S$ is stored within this iteration. Other
  iterations may only improve our solution and thus the final solution
  we output is within $2\varepsilon$ of the optimal one.\end{proof}

Unfortunately, it is unlikely that, in general, for polynomially
bounded values of $k$ and $b$, there is an FPTAS for
$(k,b)$-bribery. The reason for this is hidden in the proof of
Theorem~\ref{thm:npcom}: If there was an FPTAS for
$(1,1)$-weighted-bribery then, via using appropriately good
approximation, one could solve plurality-weighted-negative-bribery
exactly. Thus, we have the following corollary.

\begin{corollary}
  There is no FPTAS that solves $(k,b)$-weighted-bribery instances for
  any given $k,b$.
\end{corollary}

We point readers interested in the issue of approximation with respect
to control and bribery to the work of
Brelsford~\cite{bre:t:approximation}.

We conclude this section on a somewhat different note. Copeland rule,
which here we will call\footnote{We are following the naming scheme
  of~\cite{fal-hem-hem-rot:c:llull} here.} Copeland$^{0.5}$, is the
following: For each pair of candidates $c_i$ and $c_j$ we ask each
voter which one among the two he/she prefers. The one preferred by
majority receives one point. In case of a tie both candidates receive
half a point. In the end, the candidates with most points are winners.

Faliszewski et al.~\cite{fal-hem-hem-rot:c:llull}, among others,
studied similar elections where voters represent their votes via so
called irrational preference tables. An irrational preference table is
a symmetric function that given two candidates returns the one that
the voter prefers. Faliszewski et al.~\cite{fal-hem-hem-rot:c:llull}
defined a very natural bribery model for thus represented Copeland
elections, where flipping a value of each irrational preference table
entry comes at unit cost; they called this problem microbribery. We
note that microbribery in the case where flipping each entry of each
preference table may have a different price is a natural example of
what in this paper we call nonuniform bribery. Faliszewski et al. gave
polynomial-time algorithms for microbribery for Copeland$^0$ (a system
just like Copeland$^{0.5}$, but where no points are granted in case of
a tie in a head-to-head contest) and for Copeland$^1$ (defined,
analogously, as giving one point to each candidate in a tied
head-to-head contest).  Their algorithms naturally extend to a
situation where each preference-table-entry-flip comes with its own
price. Here we report that via fairly simple modifications their
algorithm can also be extended to handle Copeland$^{0.5}$.

\begin{theorem}
  There is a polynomial-time algorithm for microbribery in
  Copeland$^{0.5}$, even if each preference-table-entry-flip comes
  with its own, potentially distinct, price.
\end{theorem}

We omit the proof, which we expect to include in the full version
of~\cite{fal-hem-hem-rot:c:llull}. However, we mention that the proof
again uses the flow-network technique we used before and that the key
difference as compared to~\cite{fal-hem-hem-rot:c:llull} is that we
include a gadget that allows us to ``split'' a point that travels in
the network in case we need to model a tie between two candidates.

We mention that it does not seem to be easy to generalize our
algorithm to tie-handling values other than $0$, $\frac{1}{2}$, and
$1$. The cases of $0$ and $1$ are special because we can manage to
pump appropriate number of points through our flow network, and the
case of $\frac{1}{2}$ is easy to handle as then both candidates in a
tie are handled symmetrically in the network. We suspect that
microbribery for other tie-handling values might be $\np$-complete.

\section{Discussion and Open Problems}

In this paper we have introduced and studied the concept of nonuniform
bribery, where briber's payment to each particular voter depends on
the nature of the bribery performed. We have argued that this
generalization of the problem is interesting as it both models the
reality more precisely (voters may not be willing to modify their
votes in certain ways, or may require higher payments for such
changes), and it can be useful for coalition-formation tasks and to
model briber's policies. We have shown efficient algorithms for
nonuniform bribery for the case of plurality, veto, variants of
approval, variants of utility-based voting, and Copeland. However, our
nonuniform bribery problem is so expressive that as soon as we
consider weighted voters we can easily find special cases that are
$\np$-complete.

We have also shown that some of the known $\np$-complete bribery
problems in fact have fully polynomial-time approximation schemes.

We conclude this paper with several open questions. In particular we
are interested in the complexity of unweighted manipulation of Borda
count (an election system somewhat similar to our utility based
voting, but where the amounts of points that the voters has to assign
to candidates are set very rigidly). To the best of our knowledge,
unweighted coalitional manipulation has not been studied for Borda.
Similarly, we are interested in the complexity of unweighted bribery
for Borda. We also mention that the complexity of microbribery for
Copeland and tie-handling values other than $0$, $\frac{1}{2}$ and $1$
is unknown and poses an interesting challenge.

\section*{Acknowledgements}

I would like to thank Edith Hemaspaandra, Lane Hemaspaandra, J\"org
Rothe and Henning Schnoor for very interesting and useful
discussions.

\bibliography{grypiotr2006}
\small
\bibliographystyle{alpha}
\end{document}